\RequirePackage[2024-06-01]{latexrelease}
\documentclass[aip,cha,reprint,amsmath,amssymb,nofootinbib,floatfix]{revtex4-2}

\usepackage{graphicx,bm,booktabs,amsthm}
\usepackage[most]{tcolorbox}
\usepackage{array}
\usepackage{colortbl}
\definecolor{lbblue}{RGB}{225,236,250}
\definecolor{ubred}{RGB}{252,226,226}
\newtcolorbox{ideabox}[2][]{breakable,enhanced,colback=blue!4,colframe=blue!45!black,boxrule=0.5pt,arc=1.5pt,
  left=5pt,right=5pt,top=4pt,bottom=4pt,fontupper=\small,title={#2},fonttitle=\small\bfseries,
  coltitle=blue!45!black,colbacktitle=blue!4,titlerule=0pt,toptitle=1pt,bottomtitle=0pt,#1}
\newtcolorbox{conjbox}[1][]{enhanced,unbreakable,colback=yellow!30,colframe=red,boxrule=0.6pt,
  sharpish corners,left=5pt,right=5pt,top=3pt,bottom=3pt,fontupper=\small,title={Conjecture},
  fonttitle=\small\bfseries,center title,#1}
\newtheorem{theorem}{Theorem}
\newtheorem{lemma}{Lemma}
\newtheorem{proposition}{Proposition}

\newcommand{\Z}{\mathbb{Z}}

\begin{document}

\title{Beyond 11/16: dense Kuramoto networks that resist synchrony}

\author{Kevin P. O'Keeffe}
\email{kevin.p.okeeffe@gmail.com}
\affiliation{Starling Research Institute, Seattle, WA 98112, USA}

\date{\today}

\keywords{Kuramoto model, synchronization, complex networks, interval arithmetic}

\begin{abstract}
A dense network of identical Kuramoto oscillators usually synchronizes, but not always: some dense
networks also have stable states that are not synchronized. How dense can such networks be? The critical
connectivity $\mu_c$ beyond which they cannot exist was known to lie between $11/16=0.6875$, set by
Canale's eight-cluster networks, and $3/4$. We construct networks with stable nonsynchronous states and
connectivity arbitrarily close to $\mu_0=0.6929077\ldots$, an algebraic number of degree 8, so
$\mu_c\ge\mu_0$. The states grow out of a four-cluster state found by Townsend,
Stillman and Strogatz, which lies on what they called the razor's edge of stability: its energy is flat
to second order under rigid rotations of its clusters. We split three of its clusters into nearly coincident halves and thin the
links joining the closest and the farthest halves. Splitting lifts the flat directions, thinning pays for
it in connectivity, and $\mu_0$ is where the two balance. Because the states are only barely
stable, we prove their stability with computer checks whose rounding errors are rigorously controlled
(interval arithmetic).
\end{abstract}

\maketitle

\begin{quotation}
\noindent\textbf{Networks of identical oscillators with attractive coupling usually synchronize, but not
always: in sparse networks, other stable states can coexist with sync. For the Kuramoto model the dividing
line is a critical connectivity $\mu_c$, the fraction of the other oscillators to which each one must be
coupled before sync is guaranteed to be the only stable state. Until recently it was known only to lie between
about $0.68$ and $3/4$. Lower bounds come from explicit networks with stable nonsynchronous states, and over the past
decade they crept up from $0.6809$ to Canale's $11/16=0.6875$. We construct denser networks, with
connectivity approaching $\mu_0=0.6929\ldots$, an exact algebraic number: a root of a polynomial of degree
eight. Their stable states are small deformations of a state of four clusters a quarter turn apart, which
sits exactly at the margin of linear stability: splitting three of its clusters by small angles, and
removing a few links, makes it stable.}
\end{quotation}

\begin{figure*}[t]
\centering
\includegraphics[width=\textwidth]{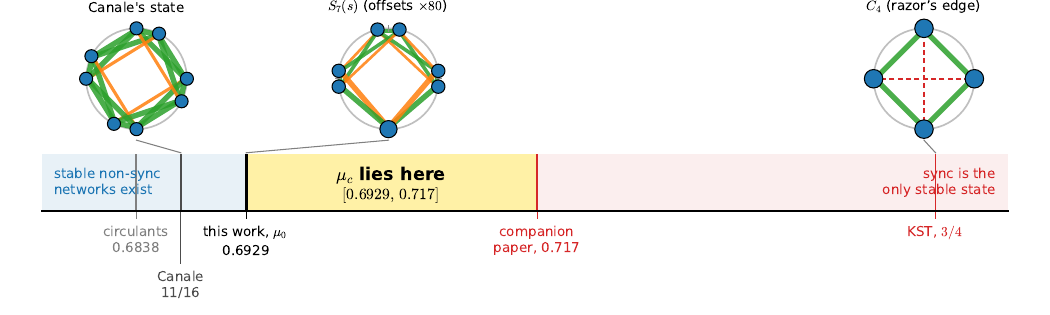}
\caption{Where $\mu_c$ stands. Below the window (blue), networks with stable nonsynchronous states are
known: the best circulants, Canale's eight-cluster state, and the split states of this paper, whose
degree approaches $\mu_0=0.6929\ldots$. Above it (red), sync is proved to be the only stable state: down to
$0.717$ by the companion paper~\cite{OKeeffeUB}, improving the $3/4$ of Kassabov, Strogatz and
Townsend~\cite{KST2021}, where the razor's-edge state $C_4$ sits. Thumbnails show the states on the phase circle: dots are clusters, with
area proportional to mass; chords are links between clusters, green when complete ($P_{ab}\ge0.999$)
and orange when partial, and thicker when they carry more link mass $x_ax_bP_{ab}$; in $C_4$ the absent links between opposite clusters
are dashed red. The offsets of $S_7(s)$, drawn at $s=0.002$ (degree $0.69272$), are exaggerated
eighty-fold.}
\label{fig:window}
\end{figure*}

\section{Introduction}\label{sec:intro}

Networks of identical, attractively coupled oscillators tend to synchronize. On the complete graph sync
is the only stable state, but on sparser networks it can fail: rings, for example, can trap the
oscillators in twisted states, whose phases wind once or more around the ring~\cite{WSG2006}. If every
oscillator is coupled to enough of the others, though, sync is again the only stable
state~\cite{Taylor2012}. The question is how many is ``enough''.

Consider the homogeneous Kuramoto model on a simple graph with adjacency matrix $A$,
\begin{equation}
\dot\theta_j=\sum_{k=1}^{n}A_{jk}\sin(\theta_k-\theta_j),\qquad j=1,\dots,n.
\label{eq:kuramoto}
\end{equation}
It is a gradient flow for the potential, or energy,
\begin{equation}
V(\theta)=\sum_{j<k}A_{jk}\bigl[1-\cos(\theta_j-\theta_k)\bigr],
\label{eq:potential}
\end{equation}
so every trajectory settles to an equilibrium~\cite{Lojasiewicz1963}, and the stable equilibria are the
local minima of $V$. Throughout this paper a \textit{stable state} means a local minimum of $V$.

Following Refs.~\onlinecite{Taylor2012,KST2021}, we measure how densely a graph is wired by its
\textit{connectivity} $\mu=\delta/(n-1)$, where $\delta$ is the minimum degree: every oscillator is
coupled to at least a fraction $\mu$ of the others. The critical connectivity $\mu_c$ is the infimum of
the numbers $\mu$ such that every graph with connectivity at least $\mu$ has no stable state other than
sync.\footnote{This is the local-minimum version of $\mu_c$, as in Refs.~\onlinecite{LZ2026,OKeeffeUB}.
A stronger, ``dynamical'' version asks for convergence to sync from almost every initial
condition~\cite{KST2021}; its threshold is at least as large. Our states are strict local minima (up to
global rotation), which attract an open set of initial conditions, so our lower bound holds for both versions.}

Figure~\ref{fig:window} shows where things stand, and Table~\ref{tab:history} gives the history. Lower bounds come from networks that resist sync: twisted
states and their relatives on circulant networks gave $0.6809$, $0.6818$, $0.6828$ and
$0.6838$~\cite{Taylor2012,WSG2006,CanaleMonzon2015,TSS2020,Yoneda2021}, and Canale then broke out of the
circulant family with an eight-cluster state that reaches $11/16=0.6875$~\cite{Canale2022}. Upper bounds
are proofs that sync must win. The best was $3/4$, by Kassabov, Strogatz and Townsend~\cite{KST2021}, who
also pointed out what stands in the way of going lower: networks with connectivity near $3/4$, found by
Townsend, Stillman and Strogatz, that lie on what they called the razor's edge of
stability~\cite{TSS2020,KST2021}. Their oscillators sit in four equal clusters a quarter turn apart, a state we
denote $C_4$ (Fig.~\ref{fig:window}, right). The energy of this state is flat to second order when its clusters are moved
rigidly, so linear stability analysis cannot decide whether it is stable. Lepsveridze and Zhang showed
that higher-order terms destabilize it and that $\mu_c<3/4$ strictly~\cite{LZ2026}, and in a companion
paper we made this quantitative, $\mu_c\le0.717$~\cite{OKeeffeUB}.

Where in this window does $\mu_c$ lie? Our own numerical searches for dense stable states kept returning
Canale's state, which made $11/16$ a natural guess for the exact value. This paper shows that the guess
is wrong.

\begin{theorem}[Main result]\label{thm:V0}
For every $\epsilon>0$ there are simple graphs, of arbitrarily large order, with connectivity greater
than $\mu_0-\epsilon$, on which the Kuramoto model~\eqref{eq:kuramoto} has a strictly stable
nonsynchronous state. Here
\begin{equation}
\mu_0=\frac{4+2q_*}{7+q_*}=0.6929077063\ldots,
\label{eq:V0}
\end{equation}
and $q_*=0.6505691669\ldots$ is the unique root in $(0,1)$ of
\begin{equation}
9q^8+72q^7+288q^6+448q^5+470q^4+88q^3-208q^2-128q-15.
\label{eq:poly}
\end{equation}
Hence $\mu_c\ge\mu_0$.
\end{theorem}
Together with Ref.~\onlinecite{OKeeffeUB},
\begin{equation}
0.6929077\ldots\le\mu_c\le0.717,
\end{equation}
which narrows the window for $\mu_c$ from $[11/16,3/4)$ to less than two-fifths of its old width
(Fig.~\ref{fig:window}).

The networks that prove the theorem carry states that hug the razor's-edge state: their phases approach those of $C_4$ (Fig.~\ref{fig:window}).
The idea behind them fits in one sentence: \emph{the razor's-edge state has plenty of links but no
stability, so trade some links for stability.} The stability comes from splitting. Three of the four
clusters are each split into two halves, sitting a small angle $s$ either side of their quarter-turn
position, and this lifts the flat directions of the energy, but only at order $s^2$, and only if some links
are removed. The links removed are the four ``crossed'' links between neighboring split clusters, joining
the closest and the farthest pair of halves, which are thinned to a density $q$. Thinning costs
connectivity, so we want to thin as little as possible. To second order in the splitting, the split states are
stable just below a critical density $q_*$ (and, numerically, not above it), and the connectivity at $q_*$ is $\mu_0$.

\begin{table*}[t]
\caption{Bounds on the critical connectivity $\mu_c$. Lower bounds (left) are networks with stable
nonsynchronous states; upper bounds (right) are theorems that sync must win. The upper bounds of
Refs.~\onlinecite{LZ2026,OKeeffeUB} concern local minima; Ref.~\onlinecite{KST2021} also proves
convergence to sync from almost every initial condition. WSG networks are the nearest-neighbor rings of
Ref.~\onlinecite{WSG2006}.}
\label{tab:history}
\small
\centering
\setlength{\tabcolsep}{4pt}
\begin{tabular}{@{\extracolsep{0pt}}>{\columncolor{lbblue}[\tabcolsep][1.5\tabcolsep]}l>{\columncolor{lbblue}[1.5\tabcolsep][1.5\tabcolsep]}l>{\columncolor{lbblue}[1.5\tabcolsep][1.5\tabcolsep]\raggedright\arraybackslash}p{0.22\textwidth}>{\columncolor{lbblue}[1.5\tabcolsep][\tabcolsep]}l|>{\columncolor{ubred}[\tabcolsep][1.5\tabcolsep]}l>{\columncolor{ubred}[1.5\tabcolsep][1.5\tabcolsep]}l>{\columncolor{ubred}[1.5\tabcolsep][1.5\tabcolsep]\raggedright\arraybackslash}p{0.22\textwidth}>{\columncolor{ubred}[1.5\tabcolsep][\tabcolsep]}l}
\hline\hline
\multicolumn{4}{>{\columncolor{lbblue}}c|}{lower bounds} & \multicolumn{4}{>{\columncolor{ubred}}c}{upper bounds}\\
year & bound & method & ref. & year & bound & method & ref.\\
\hline
2006 & $\mu_c\ge0.6809$ & twisted states on rings & \onlinecite{Taylor2012,WSG2006} &
2012 & $\mu_c\le0.9395$ & necessary condition for stable fixed points & \onlinecite{Taylor2012}\\
2015 & $\mu_c\ge0.6818$ & lexicographic product of WSG and complete graphs & \onlinecite{CanaleMonzon2015} &
2019 & $\mu_c\le0.7929$ & landscape analysis & \onlinecite{Ling2019}\\
2020 & $\mu_c\ge0.6828$ & twinning plus added edges & \onlinecite{TSS2020} &
2020 & $\mu_c\le0.7889$ & refined estimates & \onlinecite{LuSteinerberger2020}\\
2021 & $\mu_c\ge0.6838$ & optimal circulants for twisted states & \onlinecite{Yoneda2021} &
2021 & $\mu_c\le3/4$ & order parameters $\rho_1,\rho_2$ & \onlinecite{KST2021}\\
2022 & $\mu_c\ge11/16$ & eight-cluster construction & \onlinecite{Canale2022} &
2026 & $\mu_c<3/4$ & nonlinear stability, no estimate & \onlinecite{LZ2026}\\
this work & $\mu_c\ge0.6929077$ & split four-cluster states & &
2026 & $\mu_c\le0.717$ & trial rotation + certificates & \onlinecite{OKeeffeUB}\\
\hline\hline
\end{tabular}
\end{table*}

\section{From networks to macrostates}\label{sec:clusters}\label{sec:background}

Before building the new states we need three things, which this section supplies: a compact
description of a state by its clusters, the \textit{macrostate}; the conditions under which the state it
describes is stable; and the razor's-edge state $C_4$, whose failure points the way to our construction.
Readers who know the companion paper~\cite{OKeeffeUB} can skip to Sec.~\ref{sec:T6}.

\paragraph*{Macrostates.} To prove a lower bound we need a network and a stable nonsynchronous state on it, and we must choose both
at once. Searching over networks directly is hopeless: there are too many, and their number grows with
the size $n$. But the record-setting states in Fig.~\ref{fig:window} share a simple structure. Their
oscillators fall into a handful of phase-locked clusters, so each state is described by a few numbers:
how large the clusters are, where they sit on the circle, and how densely each pair of clusters is
linked. Following Ref.~\onlinecite{OKeeffeUB}, we package these numbers into one object, the
macrostate. It does not depend on $n$, and any macrostate that passes the stability conditions below
strictly can be turned back into actual networks (Sec.~\ref{sec:proof}). The search for dense stable
states thus becomes a problem in a few variables.

Figure~\ref{fig:coarse} shows how to read the macrostate off a network. Draw each oscillator as a point on
the unit circle at its phase, and join two oscillators by a chord whenever they share an edge. What
stands out are the stacks of oscillators at the same phase and the bundles of chords between them. Merge
each stack into one dot, a \textit{cluster} $a$, at phase $\varphi_a$, with mass $x_a$, the fraction of
oscillators it contains. Replace each bundle by one line, weighted by its link density
$P_{ab}=P_{ba}\in[0,1]$, the fraction of possible links that are present; inside a cluster with $E_{aa}$
links, $P_{aa}=2E_{aa}/(nx_a)^2$, which is $1-1/(nx_a)$ for a complete cluster and tends to $1$ for large
clusters, the limit used for $C_4$ and in Table~\ref{tab:S7limit}. Either way, an oscillator of cluster $a$ has on average $nx_bP_{ab}$
neighbors in cluster $b$. What remains, $T=(x,\varphi,P)$, is the macrostate.\footnote{In the first version of this paper we
called macrostates \emph{templates}.} Its \textit{degree}
\begin{equation}
\mu(T)=\min_a d_a,\qquad d_a=\sum_bx_bP_{ab},
\label{eq:degree}
\end{equation}
is the smallest cluster degree, as a fraction of all oscillators; it plays the role of the
connectivity.

\begin{figure}[t]
\centering
\includegraphics[width=\columnwidth]{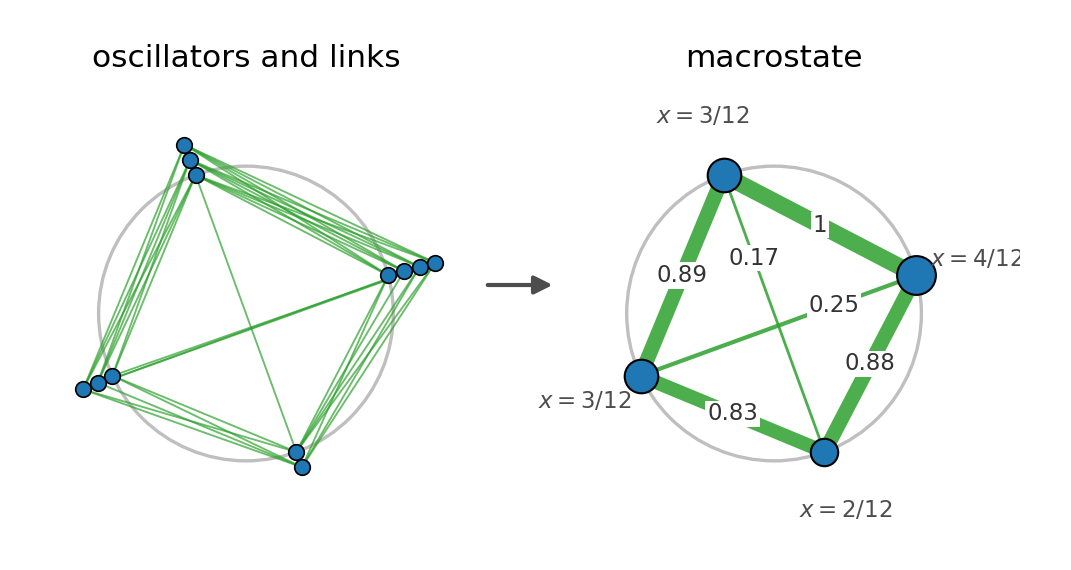}
\caption{From a network state to its macrostate. Left: twelve oscillators drawn at their
phases, stacked outward where phases coincide, with a chord for every edge. Right: each stack becomes a
dot of mass $x_a$, and each bundle of chords a line labeled by its density $P_{ab}$; links inside
stacks are omitted.}
\label{fig:coarse}
\end{figure}

\paragraph*{Stability from the macrostate.} A macrostate is useful only if we can tell from it whether the
state it describes is stable. Stability here means that the state is a local minimum of the energy $V$, and to second
order this can be tested cluster by cluster, through three conditions on $(x,\varphi,P)$:
\begin{subequations}\label{eq:three}
\begin{gather}
\textstyle F_a=\sum_bx_bP_{ab}\sin(\varphi_a-\varphi_b)=0,\label{eq:F}\\
\textstyle c_a=\sum_bx_bP_{ab}\cos(\varphi_a-\varphi_b)>0,\label{eq:c}\\
\textstyle Q(f)=\sum_{a,b}x_ax_bP_{ab}\cos(\varphi_a-\varphi_b)\,(f_a-f_b)^2>0.\label{eq:L}
\end{gather}
\end{subequations}
The first is equilibrium: the pulls on each cluster cancel. The other two are the two halves of stability,
\textit{microscopic} and \textit{macroscopic}. The microscopic condition concerns nudging a single
oscillator away from its cluster: the net pull of its neighbors, of strength $c_a$ on average, must point back toward
the cluster, like a ball in a valley rather than on a hilltop. The macroscopic condition concerns moving
whole clusters: if each cluster $a$ is rotated rigidly by $f_a$, the energy changes to second order by
$\tfrac{n^2}{4}Q(f)$, and it must go up, for every displacement $f$ that is not a global rotation
($f_a$ all equal). We call a macrostate satisfying all three conditions \textit{strictly linearly stable}.
Links inside a cluster drop out of $F_a$ and $Q$, since $\sin0=0$ and $f_a-f_a=0$; they only add to the
degree and to $c_a$.

These conditions are necessary: grouping the oscillators of any stable state by phase gives a macrostate
that satisfies them, with $\ge$ in place of $>$~\cite{OKeeffeUB}. In their strict form they are also
sufficient for our purpose. Section~\ref{sec:proof} shows that every strictly linearly stable macrostate
can be blown up into networks with minimum degree exactly $\mu(T)\,n$, on which the cluster state is
\textit{strictly stable}: a strict local minimum up to global rotation. So to prove
Theorem~\ref{thm:V0} it is enough to find strictly linearly stable macrostates with degree arbitrarily
close to $\mu_0$.

\begin{figure}[t]
\centering
\includegraphics[width=0.72\columnwidth]{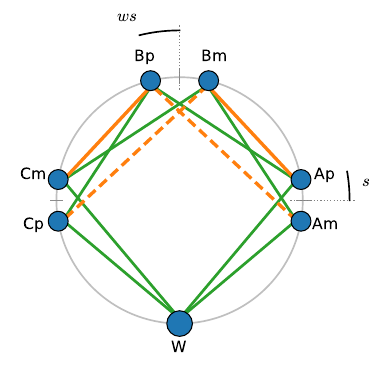}
\caption{Anatomy of the split state $S_7(s)$, splitting exaggerated. Clusters A, B and C split into halves
m and p, about $s$ (A, C) and $ws$ (B) either side of their corners; W stays whole. Green: links of
baseline density $1$. Orange: the four crossed links, of density $q\approx0.65$, joining the closest
(solid) or farthest (dashed) halves of neighboring split clusters.}
\label{fig:mechanism}
\end{figure}

The conditions also make the search tractable. For fixed masses and phases, they and the cluster degrees
are linear in the densities $P$, so finding the densest strictly linearly stable macrostate is a
semidefinite program (with a small margin in place of the strict inequalities). Our searches solve it
inside an outer loop over masses and phases (Supplementary Material). At the optimum most densities are
$0$ or $1$, with a few in between; we call this pattern of entries the \textit{link pattern}.

\paragraph*{The previous records.} The previous record-holders are easy to read in this language.
\emph{Twisted states}~\cite{WSG2006}, with equal clusters equally spaced around the circle, give the
circulant record $0.6838$~\cite{Yoneda2021}.
\emph{Canale's state}~\cite{Canale2022} has eight equal clusters on two interleaved squares, one rotated
from the other by $\arctan(1/2)\approx26.6^\circ$, linked by $0$s and $1$s plus a few half-density links
(Fig.~\ref{fig:window}, left). Each cluster is coupled to $5.5$ of the $8$ clusters, so $\mu=11/16$, and the
state is strictly linearly stable at degrees just below $11/16$.

\paragraph*{The razor's-edge state $C_4$.} A much denser candidate is $C_4$, which shows where the conditions
bite. It has masses $1/4$ at the four corners of a square, $\varphi_a=0,\pi/2,\pi,3\pi/2$, with
every link present inside each cluster and between adjacent clusters, and none between opposite clusters
(Fig.~\ref{fig:window}, right). Each cluster is linked to three of the four, so $\mu(C_4)=3/4$. Its two
neighbors pull on each cluster equally and oppositely, so $F_a=0$, and $c_a=1/4$, since only the
cluster's own members pull ($\cos(\pm\pi/2)=0$). The trouble is the macroscopic condition: $Q(f)=0$ for
\emph{every} $f$, because links between adjacent clusters carry the weight $\cos(\pi/2)=0$ and opposite
clusters are not linked. So $C_4$ satisfies every second-order condition at degree $3/4$, yet its energy
is flat when its clusters move rigidly. This is the razor's edge of
Refs.~\onlinecite{TSS2020,KST2021}; in Ref.~\onlinecite{OKeeffeUB} we call it the \textit{blind spot} of linear
stability. The cubic terms decide, and they destabilize $C_4$~\cite{LZ2026}.

\paragraph*{Square states.} $C_4$ illustrates a simple rule of thumb. In $Q$, each link between two
clusters is weighted by the cosine of their phase difference. Links between clusters a quarter turn apart
therefore cost nothing: they add to the degree without affecting stability. Links half a turn apart carry
weight $\cos\pi=-1$ and destabilize, while links between clusters at nearly the same phase stabilize. This
suggests looking for dense stable states near the \textit{square states}, whose phases all sit at the
corners of one square. The square may be rotated, and its corners may carry any masses, or none (see
Ref.~\onlinecite{OKeeffeUB} for examples); $C_4$ is the balanced member of this family. The order parameter $\rho_4=\sum_ax_ae^{4i\varphi_a}$ detects them: multiplying the
phases by four sends the corners to a single point, so the square states are exactly those with
$D_4=1-|\rho_4|^2=0$~\cite{OKeeffeUB}. The states of this paper approach them: $D_4\to0$ as the splitting
$s\to0$.

\section{The split states $S_7(s)$: beating 11/16}\label{sec:T6}

With macrostates and their stability conditions in hand, we can now build the states we seek: stable
nonsynchronous states with degree above $11/16$, the previous lower bound. We start from $C_4$. Its
degree, $3/4$, is well above that target, and it fails only at third order: every second-order condition holds, with zero to spare. Can we deform it
slightly so that it becomes stable, while keeping most of its links? Its flat directions come from links
a quarter turn long, which carry weight $\cos(\pi/2)=0$ in $Q$. To give them weight, some oscillators
must leave the corners. The simplest way to do this while keeping the pulls balanced is to split clusters
symmetrically about their corners.

\subsection{Anatomy of the split states}

Our states, which we call $S_7(s)$, have seven clusters (Fig.~\ref{fig:mechanism}). Label the corners $0,\pi/2,\pi,3\pi/2$ as A, B, C, W.
\begin{itemize}
\item Each of the three neighboring clusters A, B and C is split into a minus half (m) and a plus half
(p), sitting a small angle on either side of its corner: $\mp s$ for A and C, and $\mp ws$ for B. The
middle cluster is split wider, by a factor $w\approx1.4$ fixed below.
\item Every half has mass $a$, and the whole cluster W keeps the rest, $1-6a$. W stays whole, as an
anchor: in our searches, splitting it too did not help.
\item Links are complete inside each corner group and between neighboring groups, and absent between
opposite groups, as in $C_4$, with one exception: the four \textit{crossed} links between neighboring
split clusters, Am--Bp and Ap--Bm, and Bm--Cp and Bp--Cm, are thinned to density $q$. Between two
neighboring split clusters, they join the farthest and the closest pair of halves.
\end{itemize}
Table~\ref{tab:S7limit} lists this macrostate in full. Equilibrium fixes the masses, $a=1/(7+q)$, in the
limit $s\to0$; at finite $s$ the densities need corrections of order $s$ to keep the pulls exactly
balanced (Sec.~\ref{sec:V0proof}).

\begin{table}[t]
\caption{The split state $S_7(s)$ in the limit of small splitting. Rows and columns are the halves m
(displaced backward) and p (forward) of the split clusters A, B, C, and the whole cluster W, with
$a=1/(7+q)$. Densities are $1$ (complete), $q$ (thinned) or $0$ (absent). The last column is the
degree $d_a=\sum_bx_bP_{ab}$, with $d=(4+2q)/(7+q)$.}
\label{tab:S7limit}
\begin{ruledtabular}
\begin{tabular}{ccccccccccc}
 & $x_a$ & $\varphi_a$ & Am & Ap & Bm & Bp & Cm & Cp & W & $d_a$\\
\hline
Am & $a$ & $-s$ & 1 & 1 & 1 & $q$ & 0 & 0 & 1 & $d$\\
Ap & $a$ & $s$ & 1 & 1 & $q$ & 1 & 0 & 0 & 1 & $d$\\
Bm & $a$ & $\pi/2-ws$ & 1 & $q$ & 1 & 1 & 1 & $q$ & 0 & $d$\\
Bp & $a$ & $\pi/2+ws$ & $q$ & 1 & 1 & 1 & $q$ & 1 & 0 & $d$\\
Cm & $a$ & $\pi-s$ & 0 & 0 & 1 & $q$ & 1 & 1 & 1 & $d$\\
Cp & $a$ & $\pi+s$ & 0 & 0 & $q$ & 1 & 1 & 1 & 1 & $d$\\
W & $1-6a$ & $3\pi/2$ & 1 & 1 & 0 & 0 & 1 & 1 & 1 & $1-2a$\\
\end{tabular}
\end{ruledtabular}
\end{table}

Each cluster's degree is in the last column of Table~\ref{tab:S7limit}. Take an oscillator in Am. It
is linked to all of A (mass $2a$), to all of Bm but only a fraction $q$ of Bp ($a+aq$), and to all of W
($1-6a$), so its degree is
\begin{equation}
d=2a+a(1+q)+(1-6a)=1-a(3-q)=\frac{4+2q}{7+q}.
\label{eq:t}
\end{equation}
By symmetry every other half has the same degree; for B it is $2a+2a(1+q)$, since B is linked to both
A and C. W has more, $1-2a$, and the degree of the state is the minimum, so it is $d$. Note that it increases with $q$, from $4/7$ at $q=0$
to $3/4$ at $q=1$. At $q=1$, with $a=1/8$, nothing is thinned and every corner carries mass $1/4$: the
split state at $s=0$ is just $C_4$ with three of its clusters cut in half.

\subsection{Why splitting helps: trading links for stability}

How does splitting change the macroscopic condition~\eqref{eq:L}? It tilts the links between neighboring
groups. Pairs of halves that are now a little closer than a quarter turn carry a small positive weight
$\cos(\varphi_a-\varphi_b)$, of order $s$, which stabilizes; pairs that are a little farther carry a
small negative weight, which destabilizes. To first order in $s$ these contributions cancel exactly, on
each pair of neighboring groups (Appendix~\ref{app:proof}). What survives is second order. Once the stiff
``breathing'' modes, which rotate the two halves of a split cluster against each other, have relaxed
[Fig.~\ref{fig:modes}(c)], the cost of rotating the groups A, B and C by angles $\eta_A,\eta_B,\eta_C$
relative to W is
\begin{equation}
Q_{\rm red}(\eta)=2a^2s^2\,\hat Q(\eta)+O(s^3),
\label{eq:Qhat_main}
\end{equation}
a quadratic form in three variables whose coefficients depend on $q$, $w$ and one free first-order
parameter. Splitting can therefore stabilize the state, but only weakly, at order
$s^2$, and only if $\hat Q$ is positive definite.

Here thinning comes in (Fig.~\ref{fig:tradeoff}). Without it ($q=1$), $\hat Q$ has a negative diagonal
entry for every choice of $w$, so splitting $C_4$ alone makes it unstable already at second order
(Appendix~\ref{app:proof}). Thinning the crossed links, the most tilted ones on both sides, makes the
form positive definite below a critical density. But every thinned link lowers the
degree~\eqref{eq:t}. So there is a trade-off: the more links we keep, the denser the state, but
the less stable. We prove that $\hat Q$ is positive definite just below a
density $q_*$ at which two of its directions become neutral (Sec.~\ref{sec:limit}); numerically, no choice
of $w$ makes it positive above $q_*$ [Fig.~\ref{fig:tradeoff}(a)]. So, in the limit of small splitting, $q_*$
is the best trade within this link pattern, and the degree there is $\mu_0$.

\begin{figure*}[t]
\centering
\includegraphics[width=\textwidth]{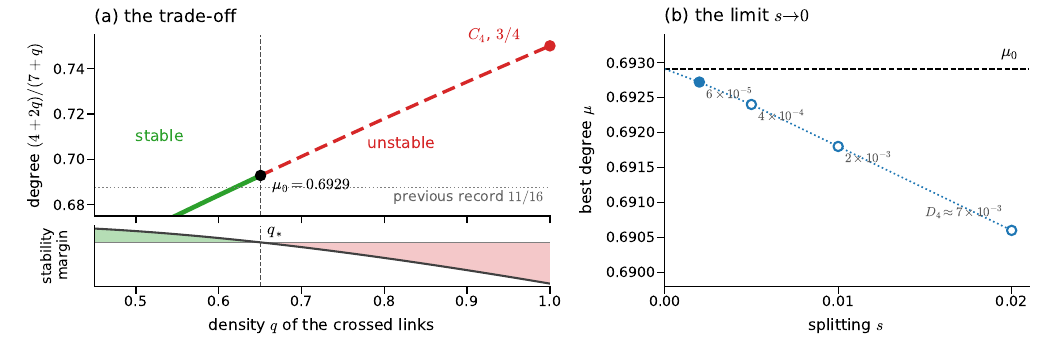}
\caption{Trading links for stability. (a) The degree $(4+2q)/(7+q)$ of the split states rises with the
density $q$ of the crossed links, up to $3/4$ at $q=1$, where the state at $s=0$ is $C_4$ with three
clusters cut in half. The strip below shows the stability margin, the smallest eigenvalue of $\hat Q$
maximized numerically over $w$ and the free first-order parameter: it is positive ($\hat Q$ positive definite) only for
$q<q_*$, so the best trade is $\mu_0$. (b) The best numerical degree at splitting $s$ (filled: the member at $s=0.002$
drawn in Fig.~\ref{fig:window}) approaches $\mu_0$ as $s\to0$; labels give $D_4$.}
\label{fig:tradeoff}
\end{figure*}

\subsection{The limit of zero splitting and the number $\mu_0$}\label{sec:limit}

Why take $s\to0$? Because the split states improve as the splitting shrinks
[Fig.~\ref{fig:tradeoff}(b)]: at finite $s$, keeping the pulls balanced costs degree at order $s$. The
limit has a second virtue: only the leading orders in $s$ survive, so the problem becomes a finite
algebraic one that can be solved exactly. As $s\to0$ the split states also approach the square states,
$D_4\to0$; in the language of Ref.~\onlinecite{OKeeffeUB}, they live in the near region, where the
upper-bound argument is hardest.

In this limit the threshold $q_*$ can be computed exactly. At $q=q_*$, with the free first-order
parameter chosen as in Appendix~\ref{app:proof}, two of the three directions of $\hat Q$ become neutral
at once and span its null space (Fig.~\ref{fig:modes}):
\begin{itemize}
\item the ``dipole'' that rotates A against C;
\item the dipole that rotates B against W.
\end{itemize}
Requiring both to be neutral gives two polynomial equations for $q$ and $w$,
\begin{align}
w^2&=\frac{q^2+6q+1}{(q+1)^2},\nonumber\\
4(1-q^2)\,w&=2q^2w^2+(3q-1)(q+1).
\label{eq:qw}
\end{align}
Eliminating $w$ leaves a single equation for $q$,
\begin{multline}
9q^8+72q^7+288q^6+448q^5+470q^4\\
+88q^3-208q^2-128q-15=0.
\label{eq:poly_sec3}
\end{multline}
Its unique root in $(0,1)$, \mbox{$q_*=0.6505691669\ldots$}, gives
\mbox{$w_*=1.3982781340\ldots$} and
\begin{equation}
\mu_0=\frac{4+2q_*}{7+q_*}=0.6929077063695863636\ldots
\label{eq:V0closed}
\end{equation}
is algebraic of degree $8$, as Eq.~\eqref{eq:poly_sec3} is irreducible. Two further facts are worth
recording. In the variable $x=(1-q)/(1+q)$ the neutrality conditions become a circle and a quartic,
$x^2+w^2=2$ and $8xw=-x^4+2x^3+x^2-8x+4$, and $\mu_0=(3+x_*)/(4+3x_*)$. And the Galois group of
Eq.~\eqref{eq:poly_sec3} is the full symmetric group $S_8$, so $\mu_0$ cannot be written in radicals.

\section{Proof of the main result: $\mu_c\ge\mu_0$}\label{sec:proof}

Section~\ref{sec:T6} found the split states and the number $\mu_0$ through an expansion in the
splitting $s$. We now show that the expansion can be made exact, which proves Theorem~\ref{thm:V0}:
there are networks with stable nonsynchronous states and connectivity arbitrarily close to $\mu_0$, so
$\mu_c\ge\mu_0$.

The proof has two halves. First (Sec.~\ref{sec:V0proof}), starting from the limit $s\to0$, which fixes $q_*$ and $w_*$, a perturbation
argument gives strictly linearly stable split states for every density $q=q_*-t$ just below threshold and
every small enough splitting $s$. Second (Sec.~\ref{sec:blowup}), a blow-up lemma turns these macrostates
into actual networks; letting $t\to0$ gives Theorem~\ref{thm:V0}. Separately (Sec.~\ref{sec:member}), one
explicit member of the family is certified by interval arithmetic and blown up in the same way, which
gives a concrete bound, Proposition~\ref{thm:main}, independent of the perturbation argument.

\subsection{Split states that are exactly stable}\label{sec:V0proof}

We start just below the
threshold, at $q=q_*-t$ with $0<t\le0.05$, and keep $w=w_*$. Only the two neutrality
conditions~\eqref{eq:qw} enter here: they place $q_*$ at the edge of stability. The proof does not need
$\mu_0$ to be the best value this family can reach, though numerically it is a local maximum at
$s\to0$. Backing off by $t$ makes the two neutral
dipoles stable, with an eigenvalue gap of order $t$. The proof (Appendix~\ref{app:proof}) has three
steps.
\begin{enumerate}
\item \emph{Ansatz.} Take the link pattern of Table~\ref{tab:S7limit} and add an explicit first-order
correction $sp'$ to the densities, plus a correction of order $s^2$ on a path of neighbor links.
\item \emph{Expansion.} Exact symbolic identities show that equilibrium holds to order $s$. After the
breathing modes are eliminated, $Q$ divided by $2a^2s^2$ tends to $\hat Q$, as in
Eq.~\eqref{eq:Qhat_main}. Interval arithmetic shows that $\hat Q$ is positive definite, with
smallest eigenvalue proportional to $t$.
\item \emph{Lifting.} For fixed $t$ and small enough $s$, the order-$s^2$ correction makes the pulls
balance exactly (by inverting a path incidence map), the densities stay in $[0,1]$, and the state is strictly
linearly stable, with degree $(4+2q)/(7+q)-O(s)$.
\end{enumerate}
The order of limits matters: fix $t$ first, then take $s$ small for that $t$, because the gap shrinks
with $t$. A separate interval computation shows that $s\le1/8$ is small enough for every
$t\in[0,0.05]$ (Appendix~\ref{app:proof}).

\subsection{From macrostates to networks}\label{sec:blowup}

Now we need actual networks. This is not automatic: a macrostate fixes only average link densities, and
in a generic graph with those densities each oscillator would feel a slightly different pull, so the
cluster phases would not even be an equilibrium. We need every oscillator in a cluster to see exactly the
same neighborhood. The idea (Fig.~\ref{fig:realization}) is to give every cluster the structure of a
ring: its oscillators sit at positions $0,1,\dots,g-1$ around a circle, with arithmetic mod $g$ (the
cyclic group $\Z_g$); in general a cluster consists of several such rings. Two rings are linked by a
fixed ``shift pattern'', as in Canale's construction~\cite{Canale2022}: oscillator $v$ of one cluster is
linked to oscillator $w$ of another whenever the offset $w-v$ lies in a prescribed set. This has three
consequences.
\begin{itemize}
\item Every oscillator sees the same pattern, so each has exactly the prescribed number of neighbors in
each cluster. The cluster phases are then an exact equilibrium.
\item Rotating all rings together is a symmetry of the network, so the Hessian splits into small blocks,
one per Fourier mode of the ring, as for twisted states on circulant networks~\cite{WSG2006}. The uniform
mode reproduces the two stability conditions~\eqref{eq:c} and~\eqref{eq:L} of the macrostate.
\item The other modes are stable as long as the shift patterns look ``random enough'', meaning their
Fourier coefficients are small. Random patterns do.
\end{itemize}

For this to work every count must be a whole number: a cluster of mass $x_a$ has $x_an$ oscillators,
and a shift set of density $P_{ab}$ has $P_{ab}g$ elements. So the lemma takes rational masses and
densities over a common denominator $D$; Lemma~\ref{lem:rational} (Appendix~\ref{app:proof}) shows that
rounding a stable macrostate to such data costs nothing.

\begin{figure}[t]
\centering
\includegraphics[width=\columnwidth]{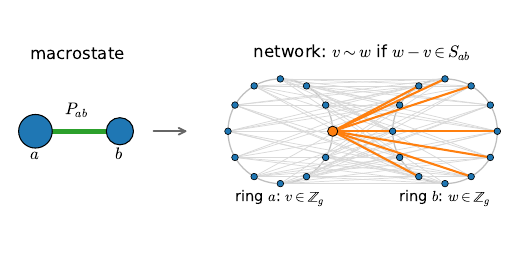}
\caption{From macrostate to network. Left: two clusters of a macrostate, joined by a line of density
$P_{ab}$. Right: each cluster becomes a ring of oscillators labeled by $\Z_g$. Oscillator $v$ of cluster
$a$ is linked to oscillator $w$ of cluster $b$ whenever $w-v$ lies in a shift set $S_{ab}$ of size
$P_{ab}g$. Orange: the links of one oscillator. One representative ring is shown per cluster. Every oscillator
has the same pattern of links, so the cluster phases form an exact equilibrium.}
\label{fig:realization}
\end{figure}

\begin{lemma}[Blow-up]\label{lem:real}
Let $(x,\varphi,P)$ be a strictly linearly stable macrostate with rational data $x_a=X_a/D$ and
$P_{ab}=Q_{ab}/D$. Here the $X_a$ and $Q_{ab}$ are integers, $D$ is even, $Q_{ab}$ is symmetric,
$0\le Q_{ab}\le D$, and $Q_{aa}\le D-1$. Let $N=m(m+3)/2$, where $m$ is the number of clusters, and
$c_{\min}=\min_ac_a$. Then for every integer $\ell\ge1$ with
$c_{\min}D\ell>2\sqrt{2D\ell\ln(8D\ell N)}+3$ there is a simple graph on $n=D^2\ell$ vertices such that
\begin{itemize}
\item each vertex of cluster $a$ has exactly $P_{ab}x_bn$ neighbors in cluster $b$, so $\delta=\mu(T)\,n$;
\item the cluster phases $\theta_j=\varphi_{a(j)}$ form a strictly stable equilibrium, which is
nonsynchronous if the $\varphi_a$ are distinct.
\end{itemize}
\end{lemma}
The condition on $\ell$ holds for every large $\ell$, since its left side grows like $\ell$ and its right
side only like $\sqrt{\ell\ln\ell}$. The proof (Appendix~\ref{app:lemmas}) follows the three points above; Hoeffding's
inequality~\cite{Hoeffding1963} shows that random shift sets have small Fourier coefficients. The
connectivity of these graphs is $\delta/(n-1)>\mu(T)$.

\begin{proof}[Proof of Theorem~\ref{thm:V0}]
Fix $\epsilon>0$. Choose $t$ so small that $(4+2q)/(7+q)>\mu_0-\epsilon/3$ at $q=q_*-t$, and then $s$ so
small that the split state of Sec.~\ref{sec:V0proof} has degree above $\mu_0-2\epsilon/3$. Approximate
its data by rationals with an even common denominator; by the implicit function theorem the state stays
strictly linearly stable, with degree above $\mu_0-\epsilon$ (Lemma~\ref{lem:rational},
Appendix~\ref{app:proof}). Lemma~\ref{lem:real} then gives the networks.
\end{proof}

\subsection{One explicit member}\label{sec:member}

Theorem~\ref{thm:V0} is an existence proof: it shows that suitable split states exist for every small
enough splitting, but does not name one. To make the result concrete, and to check it independently of
the perturbation argument, we certified one member of the family directly.

\begin{proposition}\label{thm:main}
There are simple graphs of arbitrarily large order $n$ with minimum degree exactly $\delta=\mu^\dagger n$,
where
\[
\mu^\dagger=0.6929050246493944\ldots
\]
is an explicit rational number (Supplementary Material), on which Eq.~\eqref{eq:kuramoto} has a strictly
stable nonsynchronous equilibrium. Hence $\mu_c\ge\mu^\dagger>11/16$.
\end{proposition}

This member is the output of the construction of Sec.~\ref{sec:V0proof} at $t=10^{-5}$ and
$s=10^{-6}$, less than $3\times10^{-6}$ below $\mu_0$. The difficulty is the size of the numbers: its
weakest cluster mode is only about $6\times10^{-19}$ from neutral (the smallest eigenvalue of the
cluster Laplacian with one phase pinned; a numerical estimate, while the certificate proves that it is
positive), far below double-precision roundoff, about $10^{-16}$. So we use \textit{interval arithmetic}, in which every number is replaced by an interval
guaranteed to contain it and every operation rounds outward. Two standard tools then turn numerical
evidence into proof. The \emph{Krawczyk test}, an interval version of Newton's method, proves that a small
box of phases contains exactly one equilibrium~\cite{Moore2009}. \emph{Interval Cholesky factorization}
proves that a matrix is positive definite, for every matrix in an interval box at once. Proofs of this
kind, in which a computer checks a finite list of inequalities with rigorously controlled rounding, have a
long history in dynamics: they settled the Feigenbaum conjectures~\cite{Lanford1982} and the existence of
the Lorenz attractor~\cite{Tucker2002}.

\begin{lemma}[Certificate]\label{lem:cert}
Let $(x,P)$ be the rational data in the Supplementary Material, with an even common denominator
$D\approx3.06\times10^{42}$. There is a unique vector of phases $\varphi^\dagger$ in an explicit box of
radius $10^{-30}$ (with $\varphi_1=0$) such that $(x,\varphi^\dagger,P)$ is a strictly linearly stable
macrostate with seven clusters. Its degree is exactly $\mu^\dagger$, and $P_{aa}\le1-1/D$ and $c_a\ge0.2157$ for all $a$.
\end{lemma}

The proof is in Appendix~\ref{app:lemmas}. Lemma~\ref{lem:real} applies for every $\ell\ge1$, since $c_{\min}\approx0.216$,
and Proposition~\ref{thm:main} follows. The networks are gigantic, with $n\ge D^2>9\times10^{84}$
oscillators, because this certificate uses a large common denominator. We show that these
networks exist, but we do not write one down.

\section{Discussion}\label{sec:discussion}

\paragraph*{Where is $\mu_c$?} With the companion upper bound~\cite{OKeeffeUB}, we now know
$0.6929077\ldots\le\mu_c\le0.717$. Less than two-fifths of the window $[11/16,3/4)$ that was open before
remains. Both ends are decided near the square states. The best stable states we know have phases close
to those of $C_4$, and the hardest part of the upper bound is ruling out states of this kind. We found no
state with degree above $\mu_0$, in any search we ran or any refinement of the split states we tried
(more pieces per cluster, unequal halves, splitting the anchor). So, as in Ref.~\onlinecite{OKeeffeUB}, we conjecture that the
split states are optimal.

\paragraph*{Closing the gap.} The evidence for the conjecture is local, and evidence of this kind once
supported $11/16$. Our searches started from four-cluster and Canale-type states. They also show how
easily such states are missed. The split states are barely stable: their weakest mode is only of order
$s^2(t+s)$ from neutral, while their degree falls short of $\mu_0$ by order $t+s$ (Appendix~\ref{app:proof}).
So along this family a search that demands a stability margin $\varepsilon$ loses degree like
$\varepsilon^{1/3}$, whereas
Canale's state loses only about $2.8\varepsilon$ (margins measured as in Appendix~\ref{app:lemmas}).

\paragraph*{A roadmap.} Settling $\mu_c$ now has three parts. (i) From below: search near the square
states with zero margins and exact arithmetic, over richer link patterns than ours, for anything that
beats $\mu_0$. (ii) From above: the convex certificates of Ref.~\onlinecite{OKeeffeUB}, with the
relaxation used there, cannot go below about $0.716$, so reaching $\mu_0$ needs a new idea, presumably an argument tailored to the near-square
region, like the trial rotations of Ref.~\onlinecite{OKeeffeUB} but sharp enough to see the split states,
which are stable by a hair. (iii) The dynamical threshold, for sync from almost every
initial condition, is known only to be at most $3/4$.

If the conjecture~\eqref{eq:conj} is right, ``enough'' is $0.6929077\ldots$ of the others: not a simple
fraction, nor even a number that can be written with roots, but a root of an octic whose Galois group
is $S_8$.

\begin{figure}[!t]
\begin{conjbox}[before upper={\leftskip=0pt\rightskip=0pt\parfillskip=0pt plus 1fil\parindent=0pt}]
The critical connectivity is
\begin{equation}
\mu_c=\mu_0=\frac{4+2q_*}{7+q_*}=0.6929077\ldots,
\label{eq:conj}
\end{equation}
where $q_*\approx0.6506$ is the unique root in $(0,1)$ of
\begin{multline*}
9q^8+72q^7+288q^6+448q^5+470q^4\\
+88q^3-208q^2-128q-15=0.
\end{multline*}
\end{conjbox}
\end{figure}

\section*{Methods}\label{sec:ai}
\subsection*{Use of AI tools}

AI agents (Claude Fable~5.1 and Opus~5.5, Anthropic; GPT-6-sol and GPT-6-astra, OpenAI; Grok~4.7, xAI),
guided by the author, did much of this work: they ran the searches that found the split states, derived
the proofs, wrote the two verification programs independently of each other, audited one another's
arguments, and drafted the manuscript. The author chose the problem, set the standard of evidence (exact
or interval arithmetic, independent re-verification, adversarial audits), decided which results to
trust, and edited the text. The author takes full responsibility for the content.

\subsection*{Independent checks}

The proofs use computation in two places. Theorem~\ref{thm:V0} relies on the exact symbolic
identities and one-variable interval bounds of Appendix~\ref{app:proof} (\texttt{sym\_limit.py},
\texttt{qhat\_basis.py}, \texttt{corner\_cert.py}); Proposition~\ref{thm:main} relies on the interval
certificate of Appendix~\ref{app:lemmas} (\texttt{certify.py} and a second verifier). The uniform-in-$s$
computation of the remark in Appendix~\ref{app:proof} is needed for neither. The identities are
checked in exact rational arithmetic, and the interval bounds round outward. The certificate was checked
by two independently written verifiers with different interval arithmetic, both of which reject
deliberately broken inputs. These checks enter the proofs only through the analytic lemmas of
Appendices~\ref{app:proof} and~\ref{app:lemmas}. No part has been formalized in a proof assistant.

\section*{Data availability}
The Supplementary Material contains:
\begin{itemize}
\item the exact data for the macrostate of Proposition~\ref{thm:main} (\texttt{certificate.json});
\item the exact data for the member at $s=0.002$ drawn in Fig.~\ref{fig:window}, and for a second
state, with satellite clusters, that also beats $11/16$;
\item a short note on the numerical searches.
\end{itemize}
The two verification programs and their negative controls, the scripts that check the symbolic
identities and interval bounds of Appendix~\ref{app:proof} (including the construction that
reproduces the macrostate of Proposition~\ref{thm:main} exactly, and the interval computation behind
the explicit threshold $s\le1/8$), and the logs of all these runs are available from the author on
request.

\section*{Author Declarations}
\subsection*{Conflict of Interest}
The author has no conflicts to disclose.

\appendix

\section{Proof of Theorem~\ref{thm:V0}}\label{app:proof}

This appendix turns the picture of Sec.~\ref{sec:T6} into a proof. The plan follows that section. We
write down the split state with the densities corrected to first order in the splitting $s$
(Sec.~\ref{app:ansatz}); check, by exact algebra, that its pulls balance and that its energy reduces to
the $3\times3$ form $\hat Q$ (Sec.~\ref{app:identities}); show that $\hat Q$ is positive definite just
below $q_*$ (Sec.~\ref{app:facts}); and then show that the remaining higher-order terms cannot spoil
either property for small $s$ (Sec.~\ref{app:lift}). Two short steps finish the proof: rounding the data
to rationals (Sec.~\ref{app:rational}), and choosing the parameters in the right order
(Sec.~\ref{app:conclusion}).

Each step is either argued here or is a finite computation in exact or interval arithmetic, done by a
script that we name (available from the author on request). It is convenient to use the real form
\begin{align*}
\mathcal L(y)=\tfrac12Q(y)&=\sum_{a<b}x_ax_bP_{ab}\cos(\varphi_a-\varphi_b)(y_a-y_b)^2\\
&=y^{\mathsf T}Ly,
\end{align*}
which counts each pair of clusters once; $L$ is the weighted graph Laplacian of the clusters, with
weights $x_ax_bP_{ab}\cos(\varphi_a-\varphi_b)$. Table~\ref{tab:symbols} collects the symbols of this
appendix.

\begin{table}[t]
\caption{Symbols used in Appendix~\ref{app:proof}.}
\label{tab:symbols}
\begin{ruledtabular}
\begin{tabular}{lll}
symbol & meaning & where\\
\hline
$a$, $x_W$ & masses of a half, of W & Tab.~\ref{tab:S7limit}\\
$s$, $w$ & splitting; ratio for B & Sec.~\ref{sec:T6}\\
$j_c$, $\sigma_c$ & corner and offset of cluster $c$ & A\,1\\
$q=q_*-t$ & crossed-link density & A\,1\\
$p$, $p'$, $\Delta$ & densities at orders $1$, $s$, $s^2$ & A\,1\\
$u$, $\kappa$ & free parameters of $p'$ & A\,1\\
$\mathcal L$, $L$ & $Q/2$; weighted Laplacian & above\\
$\eta$, $\zeta$ & group rotations; breathing modes & A\,2\\
$T_s$, $G_0$ & rescaled displacement; its energy & A\,2\\
$d$ & rotations across neighbor pairs & A\,2\\
$\hat Q$; $f,e,b_{33},u_{13}$ & reduced form; its entries & A\,2\\
$r_\pm$, $r_3$ & reduced minors along $w=w_*$ & A\,3\\
$G(s)$, $s_1(t)$ & rescaled energy; splitting cutoff & A\,4\\
\end{tabular}
\end{ruledtabular}
\end{table}

\subsection{The split state to first order}\label{app:ansatz}

\paragraph*{Masses and phases.} Order the seven clusters of Table~\ref{tab:S7limit} as (Am, Ap, Bm, Bp,
Cm, Cp, W). Their masses are $x=(a,a,a,a,a,a,x_W)$, with $a=1/(7+q)$ and $x_W=1-6a$, and their phases
are
\begin{gather*}
\varphi_c=j_c\frac\pi2+s\sigma_c,\\
j=(0,0,1,1,2,2,3),\qquad \sigma=(-1,1,-w,w,-1,1,0).
\end{gather*}
Here $j_c$ labels the corner of cluster $c$, and $\sigma_c$ its offset from that corner in units of $s$.
We call the clusters at one corner a \emph{group}: A, B, C or W.

\paragraph*{Densities.} Splitting tilts the links between neighboring groups, which unbalances the pulls
at order $s$. To restore the balance we correct the densities,
\begin{equation}
P(s)=p+sp'+\Delta(s).
\label{eq:Pansatz}
\end{equation}
Here $p$ is the link pattern of Table~\ref{tab:S7limit}; $sp'$ is an explicit correction that balances
the pulls to first order; and $\Delta(s)=O(s^2)$ is a small correction, constructed in
Sec.~\ref{app:lift}, that balances them exactly. The first-order correction $p'$ solves the first-order
balance equations, which leave six of the twelve neighbor-link corrections free. We take one solution
with the mirror symmetry A$\leftrightarrow$C of the state and two free parameters, $u$ and $\kappa$:
\begin{align*}
p'_{\rm Am,Bm}&=p'_{\rm Bp,Cp}=A_1-\kappa,\\
p'_{\rm Ap,Bm}&=p'_{\rm Bp,Cm}=A_2-\kappa,\\
p'_{\rm Ap,Bp}&=p'_{\rm Bm,Cm}=p'_{\rm Am,Bp}=p'_{\rm Bm,Cp}=-\kappa,\\
p'_{\rm Am,W}&=p'_{\rm Cp,W}=\tfrac{a}{x_W}(A_1+2-2\kappa),\\
p'_{\rm Ap,W}&=p'_{\rm Cm,W}=\tfrac{a}{x_W}(A_2-2-2\kappa),
\end{align*}
with $A_1=w(w+1)-u/2$ and $A_2=u/2-w(w-1)$, and $p'=0$ on all other pairs. The two
parameters have simple meanings:
\begin{itemize}
\item $\kappa$ thins all the links between neighboring groups at once. This leaves the first-order pulls
balanced, so $\kappa$ is free; we take $\kappa=19/10$, which makes $p'$ negative on every link between neighboring groups with $p=1$
[fact (C4) below; inside a group $p=1$ and $p'=0$], so that no complete link is pushed above density one.
\item $u$ shifts first-order weight between the links of neighboring groups. It enters the energy
linearly, and we fix it at the value $u_{13}(q,w)$ of Eq.~\eqref{eq:u13}, which simplifies $\hat Q$.
\end{itemize}
In summary, the ansatz is fixed by
\begin{gather*}
w=w_*,\qquad q=q_*-t\ \ (0<t\le0.05),\\
u=u_{13}(q,w),\qquad \kappa=\tfrac{19}{10}.
\end{gather*}

\subsection{Exact identities}\label{app:identities}

The ansatz was built so that the following hold identically in $(q,w,u,\kappa)$, given $a=1/(7+q)$.
They are verified in exact symbolic arithmetic by \texttt{sym\_limit.py} and \texttt{qhat\_basis.py};
every residual is $0$.
\begin{itemize}
\item[(I1)] \emph{Balance at $s=0$:} $F(p)=0$.
\item[(I2)] \emph{Balance at first order:} the order-$s$ coefficient of $F(p+sp')$ vanishes.
\item[(I3)] \emph{No first-order energy:} for each of the four pairs of neighboring groups, the order-$s$
coefficient of the total weight $\sum x_cx_dP_{cd}\cos(\varphi_c-\varphi_d)$ of the links between them
vanishes. This is the cancellation of Sec.~\ref{sec:T6}: links that are tilted closer and links that are
tilted farther balance.
\item[(I4)] \emph{The energy at second order.} The cluster modes of $S_7(s)$ are of two kinds
(Fig.~\ref{fig:modes}). \emph{Group rotations} $\eta=(\eta_A,\eta_B,\eta_C)$ turn whole groups rigidly,
with W fixed (a global rotation is not a mode). \emph{Breathing modes} $\zeta=(\zeta_A,\zeta_B,\zeta_C)$
rotate the two halves of a split cluster against each other. We scale the breathing modes by $s$, the
natural size of the splitting, and write the displacement as
\[
T_s(\eta,\zeta)=\Pi\eta+s\,z(\zeta),
\]
where $\Pi\eta$ gives both halves in a group the value $\eta_A$, $\eta_B$ or $\eta_C$ (and W the value
$0$), and $z(\zeta)=\tfrac12(-\zeta_A,\zeta_A,-\zeta_B,\zeta_B,-\zeta_C,\zeta_C,0)$. With $P=p+sp'$, the
energy $s^{-2}\mathcal L(T_s(\eta,\zeta))$ has a finite limit $G_0$ as $s\to0$. In $G_0$ the breathing
modes are stiff and decoupled from one another: their block is $a^2I$. Eliminating them leaves
$a^2\hat Q(\eta)$, where
\begin{multline}
\hat Q=u(d_0^2+d_1^2)+\bigl(\tfrac E2-u\bigr)(d_2^2+d_3^2)-(\alpha d_0-\rho d_3)^2\\
-\beta^2(d_0+d_1)^2-(\alpha d_1-\rho d_2)^2.
\label{eq:Qhat}
\end{multline}
Here $d=(\eta_A-\eta_B,\ \eta_B-\eta_C,\ \eta_C,\ -\eta_A)$ lists the relative rotations across the four
pairs of neighboring groups (A--B, B--C, C--W and W--A), and $\alpha=(w-1)-q(1+w)$,
$\beta=(1-w)-q(1+w)$, $\rho=1+q$ and $E=4(2+w^2)$. In fact $G_0/a^2=\hat Q(\eta)+|\zeta-b(d)|^2$, with
$b(d)=(\alpha d_0-\rho d_3,\ \beta(d_0+d_1),\ \alpha d_1-\rho d_2)$, so minimizing over $\zeta$ leaves
exactly $\hat Q$. Since $Q=2\mathcal L$, the reduced cost of the group rotations is
$2a^2s^2\hat Q+O(s^3)$, which is Eq.~\eqref{eq:Qhat_main}.
\end{itemize}

\begin{figure}[t]
\centering
\includegraphics[width=\columnwidth]{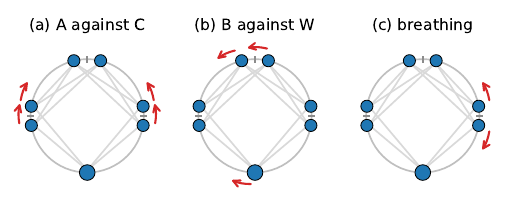}
\caption{Cluster modes of the split state $S_7(s)$, drawn as in Fig.~\ref{fig:mechanism} with arrows for
the displacement. (a) The dipole that rotates group A against group C. (b) The dipole that rotates group
B against the whole cluster W. These two group rotations are the directions in which $\hat Q$ becomes
neutral at $q=q_*$; their energies are proportional to $f+e$ and $f-e$. (c) A breathing mode, which
rotates the two halves of a split cluster against each other; such modes are stiff.}
\label{fig:modes}
\end{figure}

\paragraph*{The form $\hat Q$ in a good basis.} Since the four differences sum to zero, $d_3=-(d_0+d_1+d_2)$,
and $\hat Q$ is a form in $(d_0,d_1,d_2)$. In the basis $e_1=(1,0,-1)$, $e_2=(0,1,0)$, $e_3=(1,0,1)$ its
matrix is
\[
\begin{pmatrix}f&e&b_{13}\\e&f&b_{23}\\b_{13}&b_{23}&b_{33}\end{pmatrix},
\]
with
\begin{align*}
f&=-2q^2w^2-4q^2w-3q^2-2q+4w+1,\\
e&=-q^2w^2+q^2-2qw^2+6q-w^2+1,\\
b_{33}&=-2q^2w^2-3q^2-2q-4u+8w^2+17,
\end{align*}
$b_{13}+b_{23}=f+e$, and $b_{13}=2(u-u_{13})$, where
\begin{equation}
u_{13}=\tfrac12\bigl(2q^2w^2+2q^2w-q^2-6q+4w^2-2w+3\bigr).
\label{eq:u13}
\end{equation}
This basis makes the physics visible. The directions $e_1+e_2$ and $e_1-e_2$ are the two dipoles of
Fig.~\ref{fig:modes}(a,b), with energies $2(f+e)$ and $2(f-e)$, so the neutrality conditions of
Sec.~\ref{sec:limit} are $f=e=0$. Two consequences follow at once.
\begin{itemize}
\item \emph{Without thinning, splitting destabilizes.} At $q=1$, $f=-2w^2-4<0$ for every $w$: a diagonal
entry is negative, so $\hat Q$ is not positive definite for any $w$ and $u$.
\item \emph{A simple test for stability.} At $u=u_{13}$ we have $b_{13}=0$, and the leading minors are
$f$, $(f+e)(f-e)$ and $(f+e)\bigl[b_{33}(f-e)-f(f+e)\bigr]$. By Sylvester's criterion,
\begin{equation}
\hat Q\succ0\iff f\pm e>0\ \text{and}\ b_{33}(f-e)>f(f+e).
\label{eq:sylv}
\end{equation}
\end{itemize}

\subsection{The threshold, and stability just below it}\label{app:facts}

The following facts are proved by \texttt{corner\_cert.py}, using interval arithmetic with outward
rounding at 40 digits, and exact algebra.
\begin{itemize}
\item[(C1)] \emph{The threshold exists and is unique nearby.} The Krawczyk test shows that $f=e=0$ has
exactly one solution $(q_*,w_*)$ in a box of radius $10^{-30}$ about a rational center recorded in
\texttt{corner\_cert.py}; to 20 digits, $q_*=0.65056916694479880261$ and $w_*=1.39827813402046117964$.
\item[(C2)] \emph{It is algebraic.} The resultant of $f$ and $e$ with respect to $w$ is exactly the
polynomial~\eqref{eq:poly}, which is irreducible over $\mathbb Q$ and has exactly one root in $(0,1)$;
this root is therefore $q_*$. The degree $\mu_0=(4+2q_*)/(7+q_*)$ is the unique root in $(0.69,0.70)$ of
the irreducible polynomial
\begin{multline*}
312976\mu^8-1211776\mu^7+2547752\mu^6\\
-3537088\mu^5+3260545\mu^4-1949592\mu^3\\
+724012\mu^2-151504\mu+13636.
\end{multline*}
\item[(C3)] \emph{Stability just below the threshold.} Along $w=w_*$, $q=q_*-t$, all three quantities in
Eq.~\eqref{eq:sylv} vanish at $t=0$, because $f=e=0$ there. So $f+e=t\,r_+(t)$, $f-e=t\,r_-(t)$ and
$b_{33}(f-e)-f(f+e)=t\,r_3(t)$, with polynomials $r_\pm$ and $r_3$. Interval enclosures on $200$
subintervals of $[0,0.05]$ give $r_+\ge16.74$, $r_-\ge18.53$ and $r_3\ge281.5$. Hence $\hat Q\succ0$ for
every $t\in(0,0.05]$. Its smallest eigenvalue is at least $ct$ for some $c>0$: at $t=0$ the matrix is
$\mathrm{diag}(0,0,b_{33})$ with $b_{33}\approx15.34>0$, so for small $t$ two eigenvalues are $O(t)$ and the
third stays near $b_{33}$, while the determinant $t^2r_+r_3$ is at least a positive multiple of $t^2$;
so both small eigenvalues are at least a positive multiple of $t$, and compactness covers the rest of
the interval.
\item[(C4)] \emph{Complete links stay below one.} On $t\in[0,0.05]$, $A_1-\kappa\le-0.165$,
$A_1+2-2\kappa\le-0.065$ and $A_2-2-2\kappa\le-4.72$. With $-\kappa<0$, this gives $p'<0$ on every pair
of neighboring groups with $p=1$. Also $q\in(0.60,0.651)$.
\end{itemize}

\subsection{From first order to an exact state}\label{app:lift}

It remains to show that the higher-order terms in $s$ do no harm.

\begin{lemma}[Lifting]\label{lem:lift}
Fix $t\in(0,0.05]$. There is $s_1(t)>0$ such that for every $0<s<s_1(t)$:
\begin{enumerate}
\item[(a)] there is a unique small $\Delta(s)=O(s^2)$, supported on the six pairs of the path
Am--Bm--Cm--W--Cp--Bp--Ap, with $F(P(s))=0$ exactly;
\item[(b)] $0\le P(s)\le1$ entrywise;
\item[(c)] $c_a>0$ for every cluster, and $Q(y)>0$ unless $y$ is a global rotation;
\item[(d)] the degree is at least $(4+2q)/(7+q)-C(t)s$ for some constant $C(t)$.
\end{enumerate}
\end{lemma}

\begin{proof}
\emph{(a) Exact balance.} The pulls $F$ are affine in the densities, and they always satisfy
$\sum_cx_cF_c=0$ (action equals reaction), so the weighted pulls $(x_cF_c)$ live in the
six-dimensional space of vectors summing to zero. The six path pairs supply six adjustable densities.
Raising the density of a pair $(c,d)$ by one unit changes $(x_cF_c)$ by
$x_cx_d\sin(\varphi_c-\varphi_d)(\mathbf e_c-\mathbf e_d)$, which at $s=0$ is nonzero because every path
pair joins neighboring groups, a quarter turn apart. These six vectors are the weighted edges of a path
through all seven clusters, so they span the zero-sum space: the path can balance any pattern of pulls.
In other words, the linear map $J(s)$ from the six path densities to the pulls is invertible at $s=0$,
and hence for small $s$. By (I1) and (I2) the residual pulls $F(p+sp')$ are
$O(s^2)$, and they are analytic in $s$. So $\Delta(s)=-J(s)^{-1}XF(p+sp')$, with $X=\mathrm{diag}(x)$, is analytic and $O(s^2)$, and
$F(P(s))=0$.

\emph{(b) Densities in range.} Densities between opposite groups have $p=0$ and stay $0$. Thinned
densities equal $q+O(s)$, inside $(0,1)$. Complete densities between neighboring groups equal
$1+sp'+O(s^2)$, below $1$ for small $s>0$ by (C4). Densities inside a group stay exactly $1$.

\emph{(c) Stability.} Microscopic: as $s\to0$, $c_c$ tends to the mass of the cluster's own group, $2a$
for the halves and $x_W$ for W, because neighboring groups contribute weights $\cos(\varphi_c-\varphi_d)
=O(s)$ and opposite groups are unlinked. Macroscopic: for $s>0$, $T_s$ is an isomorphism from the group rotations and breathing
modes onto the displacements with $y_W=0$, which represent all displacements up to global rotation, so
it suffices that $G(s)=s^{-2}\mathcal L\circ
T_s$ is positive definite. The part of $G(s)$ from $p+sp'$ is analytic in $s$ and tends to $G_0$, by (I3)
and (I4). The part from $\Delta$ is $O(s)$: $\Delta$ is $O(s^2)$ and lives on pairs of neighboring
groups, whose weights $\cos(\varphi_c-\varphi_d)$ are $O(s)$, so its contribution to $\mathcal L$ is
$O(s^3)$. (This is why the path avoids opposite groups: there $|\cos|\approx1$, and an $O(s^2)$
correction would survive the rescaling by $s^{-2}$.) Hence $G(s)=G_0+O(s)$, and $G_0$ is positive
definite, because its breathing block is $a^2I$ and the remaining form is $a^2\hat Q\succ0$ by (C3).

\emph{(d) Degree.} Since $P(s)=p+O(s)$, every cluster degree is $d_c(s)=d_c(0)+O(s)$, and at $s=0$ the
degrees are those of Eq.~\eqref{eq:t}; the minimum over the seven clusters inherits the bound.
\end{proof}

\emph{Remark (an explicit threshold).} The proof does not need an explicit $s_1(t)$, but one exists:
$s_1=1/8$ works for every $t\in[0,0.05]$, even $t=0$. The reason is that the first-order correction to
$G(s)$ is itself positive definite on the plane of the two dipoles, where $\hat Q$ degenerates as $t\to0$,
so the gap is of order $t+s$ rather than $t$. We checked this with a separate interval computation on
$t\in[0,0.05]$, $0<s\le1/8$, using Taylor models in $s$ with rigorous remainders on $400$ subintervals of
$t$. It also gives $|\Delta_{cd}|\le1.135\,s^2$ for every entry and a degree of at least
$(4+2q)/(7+q)-0.9956\,s$.

\subsection{Rational data}\label{app:rational}

The blow-up of Lemma~\ref{lem:real} needs rational masses and densities. Rounding them is harmless,
because a strictly linearly stable state is robust: its phases move continuously with the data.

\begin{lemma}[Rationalization]\label{lem:rational}
Let $(x,\varphi,P)$ be a strictly linearly stable macrostate and $\epsilon>0$. There is a strictly
linearly stable macrostate $(\tilde x,\tilde\varphi,\tilde P)$ with rational $\tilde x,\tilde P$ over an
even common denominator $D$, $\tilde P_{aa}\le1-1/D$, $|\mu(\tilde T)-\mu(T)|<\epsilon$, and
$|\tilde\varphi_a-\varphi_a|<\epsilon$ for all $a$.
\end{lemma}

\begin{proof}
Fix the gauge $\varphi_1=0$, and regard the balance conditions $F_2=\dots=F_m=0$ as equations for the
remaining phases; $F_1=0$ then follows from $\sum_ax_aF_a=0$. Their Jacobian with respect to
$(\varphi_2,\dots,\varphi_m)$ is $X_{2:m}^{-1}L_{2:m,2:m}$, where $X=\mathrm{diag}(x)$, $L$ is the
weighted Laplacian defined at the start of this appendix, and the subscripts delete the first row and column. Strict linear stability makes $L_{2:m,2:m}$ positive
definite, so by the implicit function theorem the balanced phases, taken as local representatives modulo
$2\pi$, depend continuously on $(x,P)$. Now let the even integer $D$ grow,
and choose positive rational masses $\tilde x_a=X_a/D$ summing to one and symmetric rational densities
$\tilde P_{ab}=Q_{ab}/D$ in $[0,1]$, all converging to the original data; on the diagonal take
$Q_{aa}=\min\{\lfloor DP_{aa}\rfloor,D-1\}$, so that $\tilde P_{aa}\le1-1/D$ and
$|\tilde P_{aa}-P_{aa}|\le1/D$. (Diagonal densities enter neither $F$ nor $Q$.) The phases supplied by the
implicit function theorem balance the pulls exactly. Since $c_a$, the smallest nonzero eigenvalue of $L$,
and $\mu$ are continuous, strict linear stability persists, and the errors in degree and phase fall below
$\epsilon$ for large enough $D$.
\end{proof}

\subsection{Choosing the parameters}\label{app:conclusion}

Given $\epsilon>0$, the proof of Theorem~\ref{thm:V0} (Sec.~\ref{sec:proof}) chooses the parameters in
this order: first $t$, so that $(4+2q)/(7+q)>\mu_0-\epsilon/3$; then $0<s<s_1(t)$ with
$C(t)s<\epsilon/3$, so that Lemma~\ref{lem:lift} applies and the degree exceeds $\mu_0-2\epsilon/3$;
then a rational approximation by Lemma~\ref{lem:rational}, with tolerance $\min\{\epsilon/3,\gamma/3\}$,
where $\gamma$ is the smallest separation between the phases on the circle, so that the degree exceeds
$\mu_0-\epsilon$. The order matters, because the gap of
$\hat Q$ shrinks in proportion to $t$, so $s$ must be small for the chosen $t$. The resulting state is
not sync: its phases are distinct for $s>0$, and the tolerance $\gamma/3$ keeps them distinct.

The script \texttt{construct.py} carries out this construction for given $(t,s)$; with $t=10^{-5}$ and
$s=10^{-6}$ it produces the macrostate of Proposition~\ref{thm:main}.

\section{Proofs of the blow-up and certificate lemmas}\label{app:lemmas}

\subsection{The blow-up (Lemma~\ref{lem:real})}

The proof follows Sec.~\ref{sec:blowup}: build the graph, check degrees and equilibrium, split the
Hessian into Fourier blocks, and show that every block is positive definite.

\begin{proof}[Proof of Lemma~\ref{lem:real}]
Let $g=D\ell$, which is even.

\emph{Vertices.} Cluster $a$ consists of $X_a$ rings of $g$ oscillators each. A vertex is a triple
$(a,i,v)$: cluster $a$, ring $0\le i<X_a$, and position $v\in\Z_g$ on the ring. So $n=Dg$, and cluster
$a$ has $n_a=X_ag=x_an$ vertices.

\emph{Shift sets.} Choose
\begin{itemize}
\item for each pair $a<b$, a set $S_{ab}\subset\Z_g$ of size $P_{ab}g$, and put $S_{ba}=-S_{ab}$;
\item for each $a$, a set $T_a\subset\Z_g$ of size $P_{aa}g$, for links between different rings of
cluster $a$;
\item for each $a$, a symmetric set $Z_a=-Z_a\subset\Z_g\setminus\{0\}$ of size $P_{aa}g\le g-1$, for links
within a ring. Such sets exist in every size: take pairs $\{v,-v\}$, and add $g/2$ if the size is odd.
\end{itemize}

\emph{Edges.} Join $(a,i,v)\sim(b,j,w)$ for $a\neq b$ when $w-v\in S_{ab}$; join $(a,i,v)\sim(a,j,w)$
for $i<j$ when $w-v\in T_a$; and join $(a,i,v)\sim(a,i,w)$ when $w-v\in Z_a$. Each rule is symmetric
and never joins a vertex to itself, so the graph is simple.

\emph{Degrees and equilibrium.} A vertex of cluster $a$ has $X_b|S_{ab}|=P_{ab}n_b$ neighbors in cluster
$b\neq a$ and $(X_a-1)|T_a|+|Z_a|=P_{aa}n_a$ in its own cluster, so its degree is exactly $n\,d_a$ and
$\delta=\mu(T)n$. No rounding is needed, because $g$ is a multiple of $D$. With every oscillator of
cluster $a$ at phase $\varphi_a$, the right-hand side of Eq.~\eqref{eq:kuramoto} at each vertex is
$-nF_a=0$.

\emph{Fourier blocks.} At an equilibrium the Hessian of $V$ has entries
$H_{jk}=-A_{jk}\cos(\theta_j-\theta_k)$ for $j\ne k$ and $H_{jj}=\sum_kA_{jk}\cos(\theta_j-\theta_k)$.
Turning every ring by one step, $v\mapsto v+1$, is a symmetry of the network that fixes the state, so
$H$ commutes with it and splits into blocks $H_\omega$, one for each Fourier mode $\omega\in\Z_g$. Each
block is a $D\times D$ Hermitian matrix indexed by the rings $(a,i)$. Writing
$\hat S(\omega)=\sum_{v\in S}e^{2\pi i\omega v/g}$ for the Fourier coefficients of a shift set, its entries
are
\begin{itemize}
\item $nc_a-\hat Z_a(\omega)$ on the diagonal;
\item $-\hat T_a(\omega)$ between rings $i<j$ of the same cluster (and the conjugate for $i>j$);
\item $-\cos(\varphi_a-\varphi_b)\hat S_{ab}(\omega)$ between rings of different clusters.
\end{itemize}
Since the Fourier transform is unitary, it suffices to show that $H_0$ is positive definite apart from
the global rotation and that every $H_\omega$ with $\omega\ne0$ is positive definite.

\emph{The uniform mode, $\omega=0$.} Vectors that are constant on each cluster, and vectors that sum to
zero on each cluster, are both invariant under $H_0$. On the first kind, with value $y_a$ on cluster
$a$, the quadratic form of $H_0$ is $nD\,\mathcal L(y)=\tfrac{n^2}{2g}Q(y)$: this is the macroscopic
condition, positive except for the global rotation. On the second kind, $H_0$ acts on cluster $a$ as
$(nc_a-|Z_a|+|T_a|)\,\mathrm{Id}=nc_a\,\mathrm{Id}$, since $|Z_a|=|T_a|$: this is the microscopic
condition, positive because $c_a>0$.

\emph{The other modes, $\omega\neq0$.} By Gershgorin's theorem, $H_\omega$ is positive definite if every
row satisfies
\[
nc_a-|\hat Z_a|>(X_a-1)|\hat T_a|+\sum_{b\ne a}X_b|\hat S_{ab}|.
\]
If every nontrivial Fourier coefficient is at most $\beta$ in modulus, then, since $\sum_bX_b=D$ and
$n=Dg$, this holds whenever $c_ag>\beta$. A random shift set has Fourier coefficients of order
$\sqrt g$, much smaller than $g$, so random sets will do. Precisely, draw each $S_{ab}$ ($a<b$) and
$T_a$ uniformly among subsets of the right size.
\begin{itemize}
\item For $\omega\neq0$ the numbers $\cos(2\pi\omega v/g)$, $v\in\Z_g$, have mean $0$ and range at most $2$, and
so do the sines. Hoeffding's inequality for sampling without replacement~\cite{Hoeffding1963} gives
$\Pr(|\mathrm{Re}\,\hat S(\omega)|\ge\lambda)\le2e^{-\lambda^2/(2g)}$, and likewise for the imaginary
part.
\item For $Z_a$, draw half of it, $Y_a$, of size $k\le g/2-1$, uniformly from $\{1,\dots,g/2-1\}$, and
set $Z_a=Y_a\cup(-Y_a)$, adding $g/2$ if the size is odd. Then $\hat Z_a(\omega)=2\Sigma_a+\epsilon$ is
real, with $\Sigma_a=\sum_{v\in Y_a}\cos(2\pi\omega v/g)$ and $|\epsilon|\le1$. If $k=0$, then
$|\hat Z_a|\le1$. If $k>0$, the population $\{\cos(2\pi\omega v/g):1\le v\le g/2-1\}$ lies in $[-1,1]$
and sums to $-(1+(-1)^\omega)/2$, so $|\mathbb E\,\Sigma_a|\le1$, and Hoeffding's inequality gives
$\Pr(|\Sigma_a-\mathbb E\,\Sigma_a|\ge\lambda)\le2e^{-\lambda^2/(2g)}$. Outside this event,
$|\hat Z_a(\omega)|\le2\lambda+3$.
\item A union bound over the $N=\binom m2+2m$ random sets (the $S_{ab}$ with $a<b$, the $T_a$ and the
$Z_a$), the $g-1$ modes, and real and imaginary parts fails with
probability less than $4gNe^{-\lambda^2/(2g)}=1/2$ at $\lambda^2=2g\ln(8gN)$.
\end{itemize}
So with probability at least $1/2$, $|\hat S_{ab}|,|\hat T_a|\le\sqrt2\lambda$ and
$|\hat Z_a|\le2\lambda+3$, hence $\beta\le2\lambda+3<c_{\min}g$ by the hypothesis on $\ell$. For such a
choice of shift sets, $H$ is positive definite away from the global rotation, and the cluster state is
strictly stable. If the $\varphi_a$ are distinct, it is not sync.
\end{proof}

\subsection{The certificate (Lemma~\ref{lem:cert})}

\begin{proof}[Proof of Lemma~\ref{lem:cert} (computer-assisted)]
The checks in \texttt{certify.py} use exact rational arithmetic and outward-rounded interval arithmetic
at 50 digits~\cite{mpmath}. Ordinary multiprecision arithmetic only proposes the center of the box and
the preconditioner, which need not be exact.
\begin{enumerate}
\item \emph{Equilibrium.} The Krawczyk test proves that the balance equations $F_2=\dots=F_7=0$ have a
unique solution in the box of radius $10^{-30}$ about the center recorded in \texttt{certificate.json},
with $\varphi_1=0$. The remaining equation $F_1=0$ follows from
$\sum_ax_aF_a=0$ and $x_1>0$.
\item \emph{Microscopic stability.} Interval enclosures over the box give $c_a\ge0.2157$ for every $a$.
\item \emph{Macroscopic stability.} Let $M=X^{-1/2}LX^{-1/2}$, with $L$ the weighted Laplacian of
Appendix~\ref{app:proof} and $X=\mathrm{diag}(x)$. The global rotation corresponds to $v=(\sqrt{x_a})$, and
$Mv=0$ exactly for any phases. An interval Cholesky factorization of $M+vv^{\mathsf T}$ succeeds over the
whole box, so $M$ is positive definite on the complement of $v$: $Q(y)>0$ unless $y$ is a global
rotation.
\item \emph{Degree, diagonal and distinctness.} The degrees $d_a$ are computed exactly, and every
diagonal numerator equals $D-1$, so $P_{aa}\le1-1/D$ as Lemma~\ref{lem:real} requires. Interval enclosures show
that the seven phases are pairwise distinct mod $2\pi$ over the whole box, separated by at least
$10^{-6}$, so the state is not sync.
\end{enumerate}
A second verifier, written independently, confirms the result. It uses its own rational interval
arithmetic, drops a different balance equation in step 1, and replaces the Cholesky factorization in step
3 by exact tests of the leading minors.
\end{proof}

\end{document}